\documentclass[aps,pra,twocolumn,superscriptaddress,nofootinbib,longbibliography]{revtex4-2}

\usepackage{amsmath,amssymb}
\usepackage{amsthm}
\usepackage{braket}
\usepackage{graphicx}
\usepackage[colorlinks=true,allcolors=blue]{hyperref}

\newtheorem{theorem}{Theorem}
\newtheorem{lemma}{Lemma}

\newcommand{\Ctrue}{C_{\star}}
\newcommand{\xstar}{x_{\star}}
\newcommand{\mean}[1]{\langle #1 \rangle}

\newcommand{\Tr}{\mathrm{Tr}}
\newcommand{\Dtr}{D_{\mathrm{tr}}}

\begin{document}

\title{A Margolus--Levitin speed limit for observables:\\
mean energy bounds expectation-value change quadratically}

\author{Bryan Nasr}
\email{bryannasr4@gmail.com}
\affiliation{Independent Researcher}

\date{23 August 2026}

\begin{abstract}
The Mandelstam--Tamm and Margolus--Levitin quantum speed limits bound how fast a \emph{state}
evolves, using the energy variance and the mean energy above the ground state, respectively. Speed
limits on \emph{observables}---the change of an expectation value $\mean{A(t)}$---have so far used the
energy \emph{variance} (the Mandelstam--Tamm/quantum-Fisher-information lineage); the \emph{mean} energy has
not been brought to bear on the change $\Delta\mean{A}$ in a fixed state. We close this branch. First, we prove a \emph{no-go} theorem: there is no state-independent \emph{linear}
mean-energy bound on the time to change an observable's expectation by $\Delta=|\mean{A(T)}-\mean{A(0)}|$;
the optimal \emph{state-independent} mean-energy exponent of $\Delta$ is exactly two. Second, we prove the corresponding sharp
\emph{quadratic} bound,
$T\,(\mean{H}-E_0)\ge \Ctrue\,\Delta^2/\sigma_A^2$,
valid for every time-independent Hamiltonian $H$ (ground energy $E_0$), every bounded observable $A$ with
spectral spread $\sigma_A=(\lambda_{\max}-\lambda_{\min})/2$, and every pure \emph{or mixed} state, with the
dimension-independent constant $\Ctrue=1/(8\sin \xstar)=0.172506267461\ldots$, where $\xstar$ is the smallest positive
root of $\tan(x/2)=x$. The bound is tight, approached (though not attained) by a near-ground two-level family. We then promote it to
the \emph{exact} energy--time/swing trade-off curve of which $\Ctrue$ is the small-swing slope---tight at every
swing, the observable analog of the Giovannetti--Lloyd--Maccone curve for states---prove the full constant
survives for mixed states via joint convexity of the trace distance, sharpen the constant for
bandwidth-limited generators and for several observables at once, and show the quadratic law degrades to a
linear one when the initial state is an eigenvector of the observable (a sufficient condition). The result
completes the (mean-energy $\times$ observable) corner of the speed-limit landscape and, being
quadratic, is most constraining where the linear variance bounds are weakest. We discuss its cleanest physical home---autonomous quantum clocks, where it furnishes a coherent
mean-energy resolution floor complementary to the known entropy and rate bounds---and delimit honestly the
settings it does not constrain.
\end{abstract}

\maketitle

\section{Introduction}
Quantum speed limits (QSLs) quantify how fast a quantum system can change. Two foundational results anchor
the field. The Mandelstam--Tamm (MT) bound~\cite{MandelstamTamm1945} uses the energy \emph{variance}
$\Delta H=\sqrt{\mean{H^2}-\mean{H}^2}$: the time to reach an orthogonal state obeys
$T_\perp\ge \pi/(2\Delta H)$ ($\hbar=1$). The Margolus--Levitin (ML) bound~\cite{MargolusLevitin1998} uses
instead the \emph{mean} energy above the ground state, $T_\perp\ge \pi/[2(\mean{H}-E_0)]$; the two are
unified and tight together~\cite{LevitinToffoli2009}. Both, and the large literature that
followed~\cite{DeffnerCampbell2017}, bound the motion of the \emph{state} (fidelity, Bures angle,
orthogonalization).

A parallel and rapidly growing line bounds the motion of an \emph{observable}---the rate of change of an
expectation value $\mean{A(t)}=\Tr[\rho(t)A]$---which is what an experiment actually
reads~\cite{GarciaPintos2022UnifyingObservables,mohan2022quantum,Bringewatt2024GeneralizedGeometric,Carabba2022operatorflows,Cafaro2026FluctuationGrowth}.
This program now includes a first-of-kind multiparticle experiment~\cite{Miao2025Experiment}. With few
exceptions, these observable QSLs bound the expectation-value \emph{rate} through the energy \emph{variance}.
The unifying bound of García-Pintos \emph{et al.}~\cite{GarciaPintos2022UnifyingObservables} reads
$|\dot a|\le \Delta A\sqrt{I_F}$ with $I_F$ the quantum Fisher information ($I_F=4\Delta H^2$ for pure
states); Mohan and Pati~\cite{mohan2022quantum} obtain $|\mathrm d\mean{O}/\mathrm dt|\le 2\Delta O\,\Delta H$;
the geometric generalization of Bringewatt \emph{et al.}~\cite{Bringewatt2024GeneralizedGeometric} again
uses the (generalized) Fisher information. All are \emph{linear} in the observable change and powered by an
energy \emph{variance}. The \emph{mean}-energy resource that drives the Margolus--Levitin bound has so far
reached observables only as an \emph{operator}-distance bound for operator
flows~\cite{Carabba2022operatorflows} or through the mean of a transition Hamiltonian for macroscopic
transitions~\cite{Hamazaki2022SpeedLimits}---never as a bound on the change $\Delta\mean{A}$ of a bounded
observable in a fixed state. This is the gap we fill.

We prove that the mean energy controls observable change \emph{quadratically} and not linearly. Concretely
(Sec.~\ref{sec:main}), for a time-independent $H$ with $E_0=\min\mathrm{spec}(H)$, a bounded Hermitian $A$
with spectral spread $\sigma_A=(\lambda_{\max}(A)-\lambda_{\min}(A))/2$, and $\Delta=|\mean{A(T)}-\mean{A(0)}|$,
\begin{equation}
\boxed{\;T\,(\mean{H}-E_0)\;\ge\;\Ctrue\,\frac{\Delta^2}{\sigma_A^2}\;},
\qquad \Ctrue=\frac{1}{8\sin \xstar},
\label{eq:main}
\end{equation}
with $\xstar=2.331122370\ldots$ the smallest positive root of $\tan(x/2)=x$, giving
$\Ctrue=0.172506267461\ldots$. A matching no-go (Sec.~\ref{sec:nogo}) shows no \emph{linear} mean-energy
bound can exist; the quadratic law is the best possible. We then sharpen Eq.~\eqref{eq:main} to the
\emph{exact} tight trade-off curve valid at every swing (Sec.~\ref{sec:constant}). We work in units $\hbar=1$ and, without loss of
generality, shift $E_0=0$ so that $H\succeq0$; $\hbar$ is restored in the applications.

\section{The speed-limit dichotomy}
\label{sec:dichotomy}
The result completes a simple $2\times2$ classification (Table~\ref{tab:dichotomy}), organized by
\emph{what moves} (a state or an observable) and \emph{which energy functional powers the bound} (the
variance or the mean energy).

\begin{table*}[t]
\renewcommand{\arraystretch}{1.6}
\begin{ruledtabular}
\begin{tabular}{lcc}
 & resource $=$ variance $\Delta H$ (Mandelstam--Tamm) & resource $=$ mean energy $\mean{H}-E_0$ (Margolus--Levitin)\\
\hline
state, $1-F$ & $L\le \Delta H\,T$ & $1-F\le K\,T(\mean{H}-E_0)$ \quad (linear)\\
observable, $\Delta\mean{A}$ & $T\ge \Delta/(2\Delta H\,\sigma_A)$ \quad (linear) & $T(\mean{H}-E_0)\ge \Ctrue\,\Delta^2/\sigma_A^2$ \quad \textbf{(quadratic; this work)}\\
\end{tabular}
\end{ruledtabular}
\caption{The speed-limit dichotomy: \emph{what moves} (a state or an observable) $\times$ \emph{which energy
functional} powers the bound (variance or mean energy). Three corners were known; the lower-right
(mean energy $\times$ observable) corner is completed here and is the \emph{only} quadratic one.}
\label{tab:dichotomy}
\end{table*}

\noindent The mechanism behind that gap is transparent, and it is the heart of the no-go. The mean energy
bounds only how far the \emph{state} moves---a geometric quantity---through the ML inequality
$1-F\le K\,T(\mean{H}-E_0)$, $F=|\!\braket{\psi(0)|\psi(T)}\!|$. An observable change reads that geometry only
through $\Delta\le 2\sigma_A\sqrt{1-F^2}$, a \emph{square root}. Composing a square root (geometric) with a
linear (energetic) estimate yields $\Delta\propto\sqrt{T(\mean{H}-E_0)}$, i.e.\ a quadratic law. By contrast
MT bounds the observable \emph{velocity} $|\mathrm d\mean A/\mathrm dt|=|\mean{[H,A]}|\le 2\Delta H\,\sigma_A$
directly---no square root---hence linear. (Here $K=\sup_{x>0}(1-\cos x)/x=\sin\xstar$; see
Sec.~\ref{sec:main}.)

\section{No-go: no linear mean-energy bound}
\label{sec:nogo}
\begin{theorem}[No-go]
For time-independent $H$ and bounded Hermitian $A$ there is no positive state-independent constant $\kappa$
with $T(\mean{H}-E_0)\ge \kappa\,\Delta/\sigma_A$ for all $(H,A,\psi)$. More strongly, the sharp
mean-energy bound is quadratic: the exponent of $\Delta$ is exactly $2$.
\end{theorem}
\begin{proof}
By Eq.~\eqref{eq:main}, $P:=T(\mean{H}-E_0)\ge \Ctrue\Delta^2/\sigma_A^2$, and the near-ground two-level
family of Sec.~\ref{sec:constant} saturates it, $P_{\min}(\Delta)=\Ctrue\Delta^2/\sigma_A^2$ as an infimum.
Hence $P_{\min}(\Delta)/(\Delta/\sigma_A)=\Ctrue\,\Delta/\sigma_A\to0$ as $\Delta\to0$: any fixed positive
$\kappa$ is violated by near-ground two-level states. The tight exponent is $2$.
\end{proof}

\noindent The microscopic reason is that $\mean{A(t)}=\sum_{m,n}\overline{c_m}c_n A_{mn}e^{i(E_m-E_n)t}$ is a
\emph{signed} sum over Bohr frequencies weighted by coherences, in contrast to the probability-weighted,
ground-referenced sum $\chi(T)=\sum_n p_n e^{-iE_nT}$ that the mean energy controls. The coherence weight that
sets $\Delta$ and the populations that set $\mean H-E_0$ are independent, which is exactly why a naive linear
transfer of the ML state bound fails.

\section{The sharp quadratic bound}
\label{sec:main}
\begin{theorem}[Main]
Let $H$ be a time-independent Hermitian operator on a finite-dimensional Hilbert space with $E_0=0$
($H\succeq0$), and let $A$ be a bounded Hermitian observable with spectral spread $\sigma_A$. For any pure
state, with $\Delta=|\mean{A(T)}-\mean{A(0)}|$,
\begin{equation}
T\,\mean{H}\;\ge\;\Ctrue\,\frac{\Delta^2}{\sigma_A^2},\qquad \Ctrue=\frac{1}{8K},\quad K=\sup_{x>0}\frac{1-\cos x}{x}.
\end{equation}
\end{theorem}

\begin{proof}
Work in the energy eigenbasis, $\ket{\psi}=\sum_n c_n\ket n$, $p_n=|c_n|^2$. Both sides scale as
$\sigma_A^2$, and $\Delta$ and the dynamics of $\mean{A(t)}$ are invariant under $A\to A+cI$ (since
$\Tr[\rho_T-\rho_0]=0$); center $A$ so that $\|A\|_{\mathrm{op}}=\sigma_A$ and set $\sigma_A=1$.

\emph{Step 1 (geometry).} With $\rho_t=\ket{\psi(t)}\!\bra{\psi(t)}$ and $X=\rho_T-\rho_0$ (traceless,
Hermitian), Hölder's inequality gives $\Delta=|\Tr[AX]|\le\|A\|_{\mathrm{op}}\|X\|_1=\|X\|_1$. For two pure
states the nonzero eigenvalues of $X$ are $\pm\sqrt{1-F^2}$, $F=|\!\braket{\psi(0)|\psi(T)}\!|$, so
$\|X\|_1=2\sqrt{1-F^2}$ and
\begin{equation}
\Delta\le 2\sqrt{1-F^2}.
\label{eq:step1}
\end{equation}

\emph{Step 2 (energetics).} The survival amplitude is $\chi(T)=\sum_n p_n e^{-iE_nT}$ and $F\ge\mathrm{Re}\,\chi(T)$.
By the tangent-line lemma below, $1-\cos x\le Kx$ for all $x\ge0$, so
\begin{equation}
1-F\le 1-\mathrm{Re}\,\chi(T)=\sum_n p_n(1-\cos E_nT)\le K\,T\mean{H}.
\label{eq:step2}
\end{equation}

\emph{Combine.} $1-F^2=(1-F)(1+F)\le 2(1-F)\le 2KT\mean H$. Inserting into \eqref{eq:step1},
$\Delta^2\le 4(1-F^2)\le 8KT\mean H$, i.e.\ $T\mean H\ge \Delta^2/(8K)=\Ctrue\Delta^2$ (with $\sigma_A=1$).
Restoring $\sigma_A$ gives the claim. Retaining the exact geometry $1-F\ge1-\sqrt{1-\delta^2}$
($\delta=\Delta/2\sigma_A$) before the step $1-F^2\le2(1-F)$ yields the sharper non-asymptotic form
$T\mean H\ge(1-\sqrt{1-\delta^2})/K$, whose leading term is $\Ctrue\Delta^2/\sigma_A^2$ and which is strictly
stronger for every finite $\delta$ (by a factor approaching $2$ as $\delta\to1$).
\end{proof}

\noindent The energetics (Step~2)---ground-referencing $H$, the survival amplitude $\chi(T)$, and the
eigenbasis cosine estimate---follows the statistical-distance derivation of the Margolus--Levitin bound by
Jones and Kok~\cite{JonesKok2010GeometricDerivation} (corrected by Zwierz~\cite{Zwierz2012Comment}) and the
survival-amplitude inequality of Bhattacharyya~\cite{Bhattacharyya1983}; the new step is to place a bounded
observable on the left via Hölder (Step~1), converting the state-distance bound into one on $\Delta\mean{A}$.
Since $\chi(T)$ contains only populated levels, the bound also holds with $\mean H$ referenced to the lowest
\emph{populated} energy $E_{\min}^{\mathrm{supp}}\ge E_0$, a state-dependent strengthening.

\begin{lemma}[Tangent line]
$K:=\sup_{x>0}(1-\cos x)/x$ is attained at the unique stationary point $\xstar$ of $(1-\cos x)/x$ in
$(0,2\pi)$---equivalently, the smallest positive root of $x\sin x=1-\cos x\Leftrightarrow\tan(x/2)=x$---and
$K=\sin\xstar$. The line $y=Kx$ is tangent to $1-\cos x$ at $\xstar$ and lies above it for all $x\ge0$, with
equality only at $x=0$ and $x=\xstar$.
\end{lemma}

\noindent This tangent-line constant is the orthogonality ($q\to0$) member of the generalized
Margolus--Levitin construction of Giovannetti, Lloyd, and Maccone~\cite{Giovannetti2003QuantumLimits}; the
same $K\approx0.7246$ underlies the operator-flow bound of Ref.~\cite{Carabba2022operatorflows}. Applying the
theorem to $E_{\max}I-H\succeq0$ (which shares $F$, $\sigma_A$, and $\Delta$) gives the ceiling-referenced
companion $T\,(E_{\max}-\mean H)\ge\Ctrue\Delta^2/\sigma_A^2$ with the \emph{same} constant; combined,
$T\ge\Ctrue\Delta^2/[\sigma_A^2\min(\mean H-E_0,\,E_{\max}-\mean H)]$---the observable analog of the
bounded-spectrum (dual) Margolus--Levitin bound~\cite{Ness2022BoundedSpectrumQSL}, whose floor is weakest at
mid-spectrum and does not vanish as $\mean H$ grows. The proof uses only the pointwise tangent inequality and a
finite mean energy, so it extends to a separable infinite-dimensional Hilbert space with bounded $A$ and a
spectral measure of finite first moment.

\noindent Figure~\ref{fig:tangent} shows the lemma. All inequalities are elementary and exact; we verified
the resulting bound numerically over more than $1.5\times10^{6}$ sampled configuration--time points in
dimensions $2$--$6$, with zero violations (Sec.~\ref{sec:numerics}).

\begin{figure}[t]
\includegraphics[width=\columnwidth]{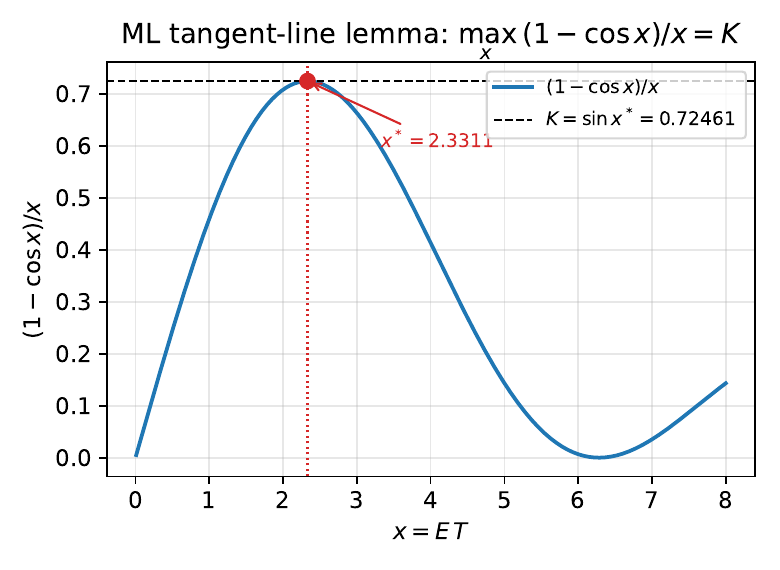}
\caption{The tangent-line lemma. $(1-\cos x)/x$ attains its supremum $K=\sin\xstar=0.724611\ldots$ at
$\xstar=2.331122\ldots$ (root of $\tan(x/2)=x$), which fixes the sharp constant $\Ctrue=1/(8K)$.}
\label{fig:tangent}
\end{figure}

\section{Exact constant, trade-off curve, and saturation}
\label{sec:constant}
The constant admits several equivalent closed forms,
\begin{equation}
\Ctrue=\frac{1}{8K}=\frac{1}{8\sin\xstar}=\frac{\xstar}{8(1-\cos\xstar)}=\frac{1+\xstar^2}{16\,\xstar},
\end{equation}
all equal to $0.172506267461\ldots$ and agreeing to $52$ digits (\texttt{mpmath}). It is dimension-independent and \emph{sharp}: it is approached,
though not attained, by the near-ground two-level family
$H=\mathrm{diag}(0,E)$, $A=\sigma_x$ ($\sigma_A=1$),
$\ket{\psi}=\cos\theta\ket0+\sin\theta\,e^{i\beta}\ket1$, in the limit $\theta\to0$ with $ET=\xstar$ and the
\emph{optimal relative phase} $\beta=\xstar/2-\pi/2$ (which starts $\mean{A(t)}$ at maximum slope rather than
at a turning point; Fig.~\ref{fig:traj}). There all three inequalities saturate simultaneously. A dedicated global optimizer over
$(E_n,A,\psi)$ in $d=2$--$10$ confirms the floor: the smallest ratio found is
$0.17250634\ge\Ctrue$, approached from above as $\theta\to0$, and no configuration dips below
(Fig.~\ref{fig:multid}). Phase restriction ($\beta=0$, real amplitudes) gives instead the strictly larger
$C_{\beta=0}=\dfrac{1}{16\,\phi_\star\sin^2\phi_\star}=\dfrac{\phi_\star}{4(1-\cos\phi_\star)^2}=0.185551\ldots$
($+7.6\%$), with $\phi_\star=2.78650\ldots$ the smallest positive root of $1-\cos x=2x\sin x$, saturated at the
turning-point configuration; the optimal relative phase recovers the full factor (Fig.~\ref{fig:converge}).

\medskip
\noindent\emph{The complete trade-off curve.} The constant $\Ctrue$ is the $\delta\to0$ slope of a sharp
\emph{curve}. With budget $P=T(\mean H-E_0)$ and normalized swing $\delta=\Delta/2\sigma_A\in[0,1)$, the
minimum budget that produces swing $\delta$---over all $(H,A,\rho,d,T)$---is
$P_\star(\delta)=\min\{T(\mean H-E_0):\Delta/2\sigma_A=\delta\}$, attained on a two-level Bohr-active subspace
and given parametrically in the Bohr phase $\tau=ET\in[\xstar,\pi)$ by
\begin{equation}
\begin{gathered}
\sin^2\theta=\frac{\tau-\tan(\tau/2)}{\tau-2\tan(\tau/2)},\qquad P_\star=\tau\sin^2\theta,\\[2pt]
\delta=\sin2\theta\,\sin(\tau/2),\qquad \tau=ET\in[\xstar,\pi).
\end{gathered}
\label{eq:curve}
\end{equation}
It interpolates from $P_\star(\delta)=4\Ctrue\delta^2[1+O(\delta^2)]$ as $\delta\to0$ (recovering the constant)
to $P_\star\to\pi/2$ as $\delta\to1$ (Fig.~\ref{fig:curve}), lying strictly above the elementary bound
$(1-\sqrt{1-\delta^2})/K$ of Sec.~\ref{sec:main}---by $1.3\%$ at $\delta=0.5$, $4.2\%$ at $\delta=0.8$, $8.2\%$
at $\delta=0.95$, but under $1\%$ in the near-ground window where the applications operate. The curve is the
\emph{observable image} of the Giovannetti--Lloyd--Maccone state trade-off: $\Delta=2\sigma_A\sqrt{1-F^2}$ is
saturable (align the $\pm\sigma_A$ eigenvectors of $A$ with the rank-two $\rho_T-\rho_0$), so the maximal swing
at fixed budget is $2\sigma_A\sqrt{1-F_{\min}^2}$ with $F_{\min}$ the minimal survival fidelity at fixed mean
energy~\cite{Giovannetti2003QuantumLimits}, whose two-level
optimizer~\cite{Hornedal2023MargolusLevitinArbitraryFidelity} carries over; a multi-start search over
$d=2$--$6$ finds no state below $P_\star(\delta)$ at any $\delta$. This is the observable counterpart---tight at
every $\delta$, not only in the constant---of the exact arbitrary-fidelity Margolus--Levitin curve for states,
whose defining equality of the two branch exponents was established in Ref.~\cite{Chau2023MLalpha}.

\medskip
\noindent\emph{Bandwidth-resolved constant.} If $\mathrm{spec}(H)\subset[E_0,E_0+B]$, the tangent-line
supremum runs only over $x\in(0,BT]$, so $K\to K(L)=\sup_{0<x\le L}(1-\cos x)/x$ with $L=BT$. Since
$(1-\cos x)/x$ increases on $(0,\xstar)$, $K(L)=(1-\cos L)/L$ for $L<\xstar$ and $K(L)=K$ otherwise, giving the
spectrally sharp constant
\begin{equation}
\Ctrue(L)=\frac{1}{8K(L)},\qquad T(\mean H-E_0)\ge \Ctrue(BT)\,\frac{\Delta^2}{\sigma_A^2},
\label{eq:bandwidth}
\end{equation}
which recovers $\Ctrue$ for $L\ge\xstar$ and is strictly stronger for narrow-band, short-time operation
($1.58\times$ at $L=1$, $2.96\times$ at $L=0.5$), saturated by a two-level gap at the band edge---sharpening
the floor for precisely the low-bandwidth devices (e.g.\ minimal clocks) where the bound is most relevant.

\begin{figure}[t]
\includegraphics[width=\columnwidth]{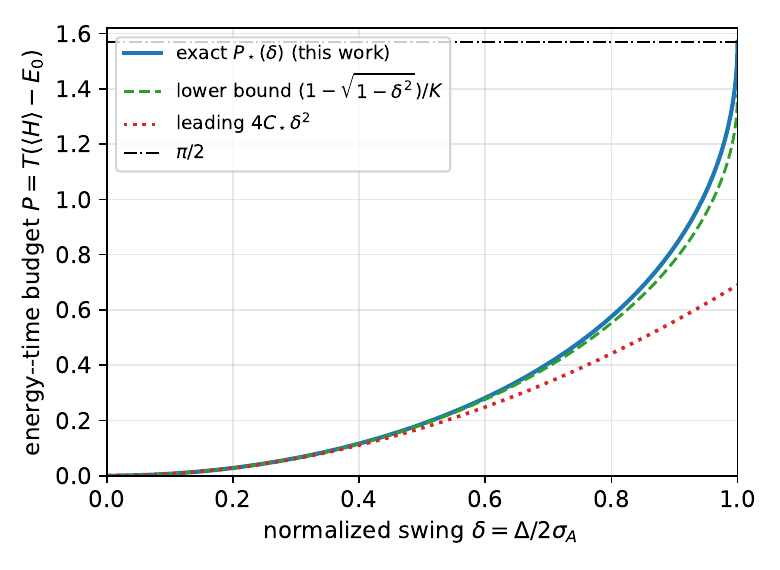}
\caption{The exact observable trade-off curve $P_\star(\delta)$, Eq.~\eqref{eq:curve}: the minimum energy--time
budget $P=T(\mean H-E_0)$ for a normalized swing $\delta=\Delta/2\sigma_A$, tight at every $\delta$. Its slope
as $\delta\to0$ is set by $\Ctrue$ and its $\delta\to1$ limit is $\pi/2$. It lies above the elementary lower
bound $(1-\sqrt{1-\delta^2})/K$ and its leading term $4\Ctrue\delta^2$.}
\label{fig:curve}
\end{figure}

\begin{figure}[t]
\includegraphics[width=\columnwidth]{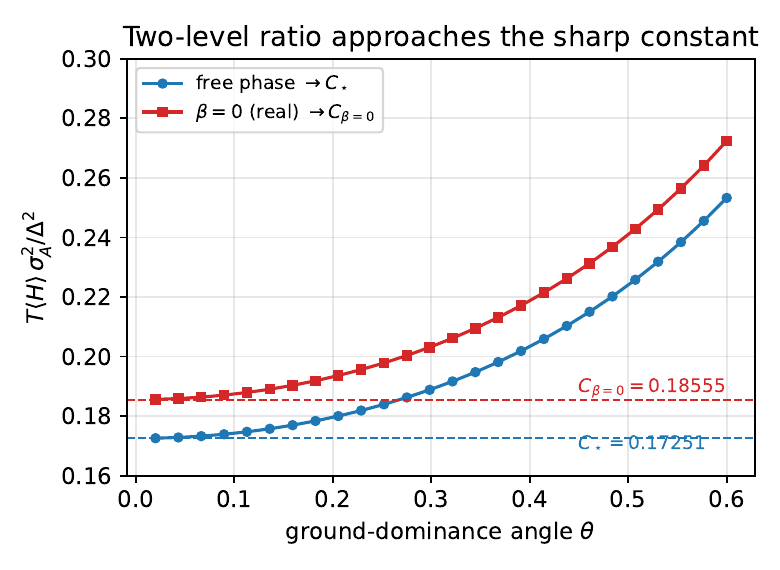}
\caption{The two-level ratio $T\mean{H}\sigma_A^2/\Delta^2$ versus the ground-dominance angle $\theta$,
approaching the sharp constant $\Ctrue$ (free relative phase) and the phase-restricted suboptimum
$C_{\beta=0}=0.185551$, both from above as $\theta\to0$.}
\label{fig:converge}
\end{figure}

\begin{figure}[t]
\includegraphics[width=\columnwidth]{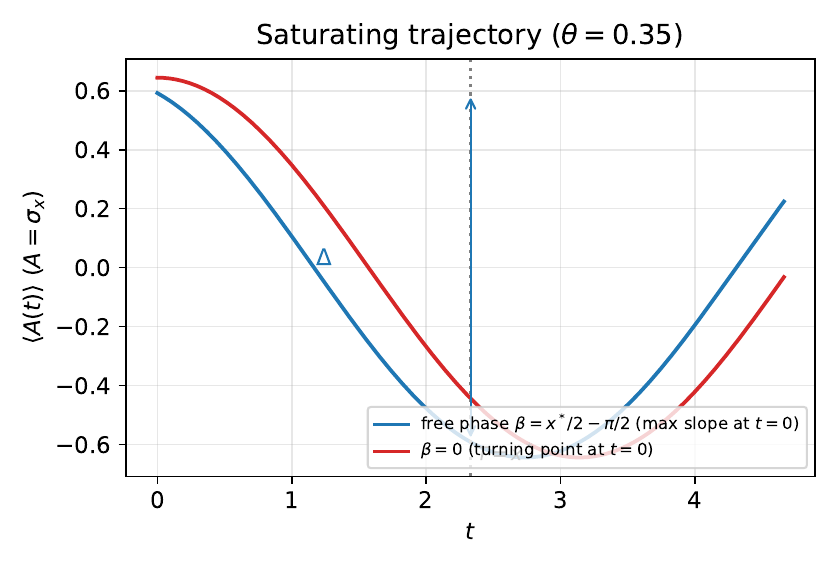}
\caption{Saturating trajectory $\mean{A(t)}$ ($A=\sigma_x$). With the optimal phase the expectation starts at
maximum slope (and reaches its excursion $\Delta$ at $T=\xstar/E$); the real-amplitude ($\beta=0$) case starts
at a turning point and is suboptimal.}
\label{fig:traj}
\end{figure}

\begin{figure}[t]
\includegraphics[width=\columnwidth]{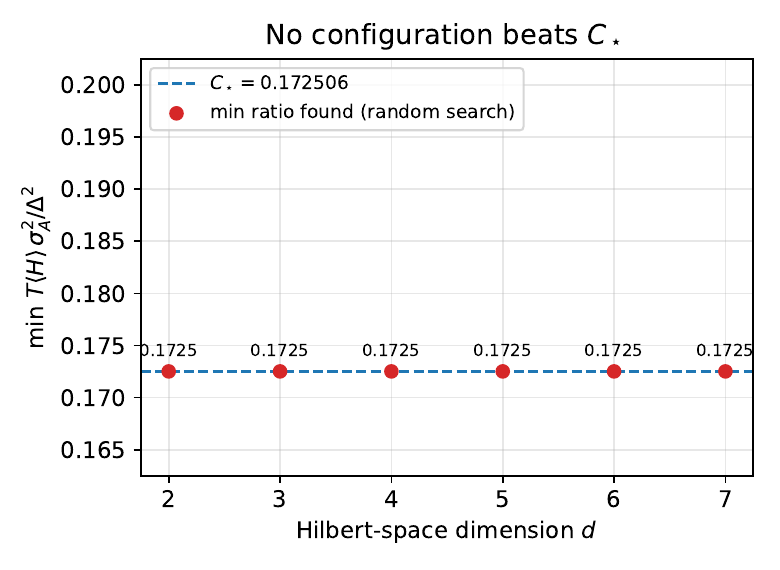}
\caption{Global-optimizer minimum of $T\mean{H}\sigma_A^2/\Delta^2$ per Hilbert-space dimension $d$.
No configuration in $d=2$--$7$ beats $\Ctrue$; each dimension approaches it from above.}
\label{fig:multid}
\end{figure}

\section{Mixed states and normalization}
\label{sec:mixed}
The bound holds for an arbitrary density matrix $\rho_0$ \emph{with the same constant} $\Ctrue$, with
$\rho_T=U\rho_0U^\dagger$, $U=e^{-iHT}$, $\Delta=|\Tr[A(\rho_T-\rho_0)]|$, $\mean{H}=\Tr[\rho_0H]$.
The geometric step generalizes via Hölder to $\Delta\le 2\sigma_A\Dtr(\rho_0,\rho_T)$, with $\Dtr$ the trace
distance. The key is to avoid an Uhlmann-fidelity detour (which loses a factor of two) by using
\emph{joint convexity} of the trace distance on the spectral decomposition $\rho_0=\sum_k p_k\ket k\!\bra k$:
since the \emph{same} $U$ evolves every eigenvector, $\rho_T=\sum_k p_k\,U\ket k\!\bra k U^\dagger$ is a convex
combination with the \emph{identical} weights $p_k$, so
\begin{equation}
\Dtr(\rho_0,\rho_T)\le\sum_k p_k\sqrt{1-F_k^2},\qquad F_k=|\!\braket{k|U|k}\!|.
\end{equation}
Per-eigenstate ML gives $1-F_k\le K T\,e_k$ with $e_k=\braket{k|H|k}\ge0$, and Jensen's inequality (concavity
of $\sqrt{\cdot}$, $\sum_k p_k e_k=\mean H$) yields
$\sum_k p_k\sqrt{2KTe_k}\le\sqrt{2KT\mean H}$. Combining,
$\Delta\le 2\sigma_A\sqrt{2KT\mean H}$, hence $T\mean H\ge\Ctrue\Delta^2/\sigma_A^2$. The same constant
follows independently by purification ($H\!\otimes\!I$, $A\!\otimes\!I$ preserve $\mean H$, $\sigma_A$ and
$\Delta$ exactly, so the pure bound descends). Each inequality is strict away from the pure $\theta\to0$ limit,
so every state with $\Delta>0$ lies strictly above $\Ctrue$; yet $\Ctrue$ remains the tight infimum for mixed
states as well (no pure/mixed gap), approached by mixing the saturating family with a vanishing ground-state
weight, consistent with the saturation analysis of Ref.~\cite{Sonnerborn2026MLsaturate}.

Among candidate normalizations, the spectral spread $\sigma_A=(\lambda_{\max}-\lambda_{\min})/2$ is canonical:
it is the Chebyshev radius $\min_c\|A-cI\|_{\mathrm{op}}$, which makes the Hölder step tight; it is
shift-invariant and state-independent. The operator norm $\|A\|_{\mathrm{op}}\ge\sigma_A$ gives a weaker valid
bound, while the state variance $\Delta A_\rho\le\sigma_A$ (Popoviciu) is the variance scale of the MT family
and is not delivered by this proof.

\medskip
\noindent\emph{Several observables at once.} Because the bound depends on the state only through
$X=\rho_T-\rho_0$, it constrains a whole family $\{A_i\}$ jointly: for any real combination $A_u=\sum_iu_iA_i$
with swing $\Delta_u=|\Tr[A_uX]|$ and Chebyshev radius $\|A_u\|_{\mathrm{Cheb}}$, Step~1 gives
$\Delta_u\le\|A_u\|_{\mathrm{Cheb}}\,2\sqrt{1-F^2}$, hence
\begin{equation}
T(\mean H-E_0)\ge \Ctrue\,\sup_{u}\frac{\Delta_u^2}{\|A_u\|_{\mathrm{Cheb}}^2},
\label{eq:multi}
\end{equation}
a mean-energy bound set by the operator-norm (Chebyshev) \emph{dual} geometry, in contrast to the Fisher-metric
multiparameter variance limits~\cite{Hamazaki2023QuantumVelocityLimits,Bringewatt2024GeneralizedGeometric}. For
two orthonormal observables (e.g.\ $\sigma_x,\sigma_y$) it reads
$T(\mean H-E_0)\ge\Ctrue(\Delta_1^2+\Delta_2^2)/\sigma_A^2$, twice the single-observable floor when
$\Delta_1=\Delta_2$.

\section{Structure: linear recovery and structured observables}
\label{sec:structure}
The quadratic law is a property of \emph{generic} observables, whose expectation starts at maximum slope. It
degrades to a \emph{linear} law when the initial state is an eigenvector of $A$ (a sufficient condition).
\begin{theorem}[Linear recovery]
If $A\ket{\psi_0}=\lambda\ket{\psi_0}$, then with $l=(\lambda-\mathrm{center}(\mathrm{spec}\,A))/\sigma_A\in[-1,1]$,
\begin{equation}
T\,(\mean{H}-E_0)\ge\frac{1}{2K(1+|l|)}\,\frac{\Delta}{\sigma_A},
\end{equation}
sharpest at an extreme eigenvalue ($|l|=1$): $T(\mean{H}-E_0)\ge \Delta/(4K)=2\Ctrue\,\Delta/\sigma_A$.
\end{theorem}
\begin{proof}
Center $A$, $\sigma_A=1$; set $B=A-\lambda I$, so $B\ket{\psi_0}=0$. Writing
$\ket{\psi_T}=e^{i\alpha}F\ket{\psi_0}+\sqrt{1-F^2}\ket{\phi}$ with $\ket\phi\perp\ket{\psi_0}$, the cross
terms vanish ($B\ket{\psi_0}=0$, $B$ Hermitian), so $\Delta=(1-F^2)|\!\braket{\phi|B|\phi}\!|\le(1+|l|)(1-F^2)$
---linear in the fidelity deficit. With $1-F^2\le2(1-F)\le2KT\mean H$ the claim follows.
\end{proof}
\noindent Thus a fidelity-aligned observable such as the projector $\ket{\psi_0}\!\bra{\psi_0}$ literally
reduces the bound to the linear ML \emph{state} bound (then $\mean{A(t)}=F^2$, $\Delta=1-F^2$), recovering the
upper-right corner of the dichotomy. We emphasize that a turning point alone (zero initial slope without
$\ket{\psi_0}$ being an eigenvector) is \emph{not} sufficient: e.g.\ $\ket{\psi_0}=(\ket0+\ket1)/\sqrt2$,
$A=\mathrm{diag}(1,-1,0)$, $H$ with real $H_{01}$ has zero initial slope yet no positive linear constant
($r_{\mathrm{lin}}$ falls to $\approx0.03\ll1/4K$).

For structured observables the constant does not improve generically: a rank-one ground coupling
$A=\ket0\!\bra v+\ket v\!\bra0$ still gives $\inf=\Ctrue$ (the saturating family is itself of this form), and
an energy-banded ground-touching $A$ likewise. Restricting to real couplings and amplitudes raises the
constant to $C_{\beta=0}=0.185551$. A ground-decoupled observable ($A_{0n}=0$) conserves the excited
population and makes the time bound vacuous; a static energy floor $\mean H-E_0\ge(\Delta_{\mathrm{gap}}/2)\Delta/\sigma_A$
governs instead. The bound is informative for observables of finite spectral spread (spins, parity,
projectors, qubit registers); for unbounded observables ($\sigma_A=\infty$, e.g.\ position or particle number)
it is vacuous, and Mandelstam--Tamm, which uses the finite state variance, is operative.

\section{Numerical verification}
\label{sec:numerics}
All numerics are seeded and reproducible (\texttt{numerics/}; consolidated in
\texttt{results.json}). The named constant is pinned to $52$ digits with \texttt{mpmath}. The no-go exponent
is confirmed by the log--log slope of $P_{\min}(\Delta)$ along the saturating family ($2.017$, against the
predicted $2$). The pure-state bound survives a $1{,}555{,}984$-point random search in $d=2$--$6$
(worst-case slacks $\ge0$ to machine precision) and a dedicated multi-start optimizer
(Fig.~\ref{fig:multid}). At $d=2$ the coarse first-passage evaluator inside that optimizer is the least
accurate step, because the configurations that come closest to the floor are near-ground ones (excited
population $\sim10^{-11}$, swing $\sim10^{-5}$); re-evaluating the exact two-level ratio there with the
first passage solved analytically in $60$-digit arithmetic returns $0.1725103\ge\Ctrue$, so the apparent
sub-$\Ctrue$ minima are an artifact of that evaluator's time grid and not violations. The
mixed-state bound is confirmed by a violation search over random density matrices (minimum sampled ratio
$0.176\ge\Ctrue$) and by a direct audit of the convexity chain; the infimum over mixed states is again
$\Ctrue$, approached as the state purifies onto the saturating family. The structured constants $\Ctrue$
and $C_{\beta=0}$ are reproduced to eight significant digits, and the linear-recovery constant $1/4K$ is
approached to within $6\times10^{-5}$. The exact trade-off curve $P_\star(\delta)$ is confirmed as a strict lower boundary
by a multi-start search in $d=2$--$6$ (no state dips below it at any $\delta$); the bandwidth-resolved bound
$\Ctrue(BT)$ by a violation search over spectrally bounded $H$ ($0$ violations in $8\times10^5$ trials); and
the two-observable bound saturates the $2\times$ gain at $\Delta_1=\Delta_2$ ($0$ violations in
$4\times10^6$ trials).

\medskip\noindent\emph{Data availability.} The manuscript source, the figure and verification scripts, and
\texttt{results.json} are openly available, archived at Zenodo~\cite{repo} and developed at
\url{https://github.com/bryannasr4-gif/observable-margolus-levitin}.

\section{Applications and discussion}
\label{sec:apps}
The most important consequence is structural: Eq.~\eqref{eq:main} is the first mean-energy, quadratic-in-$\Delta$
bound on the change of an expectation value in a fixed state, the variance/Fisher members being \emph{linear} in
$\Delta$~\cite{GarciaPintos2022UnifyingObservables,mohan2022quantum,Bringewatt2024GeneralizedGeometric}.
Being quadratic, it is most constraining exactly where the linear bounds are weakest---for \emph{large}
fractional changes $\Delta/\sigma_A$---so the two families are complementary, not competing. Observable QSLs
are an active, growing area, including recent experiments~\cite{Miao2025Experiment}.

\paragraph{A regime where the bound binds.}
There is an identifiable window in which Eq.~\eqref{eq:main} is the operative observable limit. For the
near-ground qubit $H=\tfrac\omega2\sigma_z$, $A=\sigma_x$,
$\ket{\psi_0}=\sqrt{1-\varepsilon}\,\ket g+\sqrt\varepsilon\,\ket e$ (excited population $\varepsilon$, ground
state $\ket g$) evaluated at the half-period $T_\star=\pi/\omega$, the mean-energy floor exceeds the
Mandelstam--Tamm observable floor $\Delta/(2\Delta H\sigma_A)$ by the exact factor $8\Ctrue(1-\varepsilon)$,
which is $>1$ for ground dominance $\varepsilon<0.275$. There---small excited population, large fractional
swing---the quadratic mean-energy bound, not the variance bound, sets the limit. Read in reverse,
an observed swing $\Delta$ over time $T$ \emph{certifies} a mean energy $\mean H-E_0\ge\Ctrue\Delta^2/(\sigma_A^2T)$
from a single bounded-observable expectation, with no state-overlap measurement. With $\hat\Delta$ estimated
from finite samples (standard error $s$), error propagation gives the one-sided certificate
$\mean H-E_0\ge\Ctrue[\max(0,\hat\Delta-z_\alpha s)]^2/(\sigma_A^2T)$ at confidence $1-\alpha$---a
model-independent energy floor measurable on the platforms now realizing observable speed
limits~\cite{Miao2025Experiment}, in the regime above where it binds. This is the observable-side analog of
using a speed limit to infer the energy of an otherwise opaque quantum
device~\cite{ishida2025energyinference}.

\paragraph{Autonomous quantum clocks (primary).}
An autonomous clock is by definition driven by a fixed, time-independent
Hamiltonian~\cite{Erker2017AutonomousClocks,Woods2021AutonomousTicking}, so the hypotheses of
Eq.~\eqref{eq:main} hold exactly. Taking $A$ to be the clock's bounded pointer (hand) observable and $\Delta$
a resolvable advance, the bound becomes a coherent \emph{resolution floor},
\begin{equation}
T_{\mathrm{res}}\;\ge\;\hbar\,\Ctrue\,\frac{\Delta^2}{(\mean{H}-E_0)\,\sigma_A^2},
\end{equation}
i.e.\ the mean energy stored in the clockwork lower-bounds the time to move the hand by a resolvable amount.
A precessing pointer ($H=\tfrac\omega2\sigma_z$, $A=\sigma_x$) obeys it, with the bound tightest for the
smallest, lowest-energy ``minimal'' clock; saturation is a near-ground limit, and at its natural (equatorial)
operating point the pointer is an $A$-turning point---the linear regime of Sec.~\ref{sec:structure}---so it
illustrates rather than saturates the quadratic floor. This complements---rather than subsumes---the
established clock limits, which use \emph{different} resources: dissipated entropy per tick
($N\lesssim\Delta S_{\mathrm{tick}}/2k_B$)~\cite{Erker2017AutonomousClocks,Pearson2021Timekeeping} and the
thermalization rate ($N\le\Gamma^2/\nu^2$)~\cite{Meier2023AccuracyResolution}, both for open dissipative
clocks. Mean energy alone does not \emph{upper}-bound clock or computational speed~\cite{Jordan2017FastComputation};
ours is a complementary lower bound at fixed pointer contrast $\sigma_A$.

\paragraph{Autonomous quantum batteries (secondary).}
For an autonomous (time-independent generator) battery with $A=H_B$ the stored energy, $\Delta=W$ the energy
deposited, and $\sigma_B$ the battery's spectral spread, Eq.~\eqref{eq:main} gives a minimum charging time
$T_{\min}=\hbar\Ctrue W^2/[(\mean{H_{\mathrm{tot}}}-E_0)\sigma_B^2]$, quadratic in $W$. This is a
\emph{complementary} mean-energy floor: mainstream charging is driven by time-dependent control and its
operative power bounds are variance/Fisher
based~\cite{GarciaPintos2020BatteryPower,JuliaFarre2020Battery,Campaioli2024Colloquium} (the García-Pintos
bound has a published correction~\cite{Cusumano2021Comment}), and the reported quantum advantage stems from
super-extensive \emph{variance}, not mean energy. The same algebra rearranges into an energetic-cost floor on
changing a logical observable in any always-on (autonomous) architecture; it is, however, loose for typical
fast gates, and standard time-dependent gate pulses fall outside the hypotheses. Because $\sigma_B$ tracks the
battery bandwidth, this floor falls as $1/N$ for an $N$-cell battery and sits a fixed factor below the
Margolus--Levitin charging time even at $N=1$, and it bounds the deposited energy $W$ rather than the
ergotropy~\cite{Shrimali2024StrongerSpeedLimit}; we therefore offer it as an illustrative remark, not a
binding constraint.

\paragraph{What the bound does not constrain.}
Honesty about scope sharpens the claim. (i) Standard, time-dependently driven quantum gates violate the
time-independence hypothesis---essential here, since the Margolus--Levitin bound provably does not extend to
closed systems with time-dependent generators~\cite{Hornedal2023ClosedSystems,Hornedal2023MargolusLevitinArbitraryFidelity}. (ii) For extensive many-body systems the total mean energy grows like the
system size, so the bound is weak for \emph{local} observables; there the operative limit is locality-based
(Lieb--Robinson)~\cite{Chen2023LocalitySpeedLimits}. (iii) In quantum metrology/sensing the operative speed
limit is the variance (MT) one and practical readout times are set by signal-to-noise and decoherence far
above any mean-energy floor~\cite{HerbDegen2024Sensing}; the bound nonetheless lower-bounds the mean-energy
\emph{cost} of generating a given signal swing, complementing the asymmetry/weak-value observable speed limit
of Budiyono \emph{et al.}~\cite{Budiyono2026QSLObservables}. The bound is therefore a statement about
energy-constrained, few-body or near-ground dynamics.

\paragraph{Outlook.}
Natural extensions are open-system (Lindbladian), non-Hermitian~\cite{Nishiyama2025speedlimits,wang2026general}
and time-dependent-$H$ generalizations---where much of the driven-application landscape lives---and a sharp
constant for ergotropy rather than stored energy. Within the
present setting, the constant sharpens under a spectral-bandwidth restriction
($K\to\sup_{0<x\le BT}(1-\cos x)/x$) and generalizes to higher energy moments~\cite{ChauZeng2024UnifyingQSL}, and the bound assembles with the
observable Mandelstam--Tamm limit into a unified envelope
$T\ge\max[\,\Delta/(2\Delta H\sigma_A),\,\Ctrue\Delta^2/((\mean H-E_0)\sigma_A^2)\,]$ in the spirit of
Levitin--Toffoli~\cite{LevitinToffoli2009} and the bounded-spectrum treatment~\cite{Ness2022BoundedSpectrumQSL}.

\section{Conclusion}
We have established the missing mean-energy speed limit for observables: a no-go against any linear law and a
sharp quadratic bound $T(\mean{H}-E_0)\ge\Ctrue\Delta^2/\sigma_A^2$ with a closed-form, dimension-independent
constant, valid for pure and mixed states, tight, and degrading to a linear law for eigenvector observables.
Beyond the constant we gave the complete, tight trade-off curve $P_\star(\delta)$ (the observable analog of the
arbitrary-fidelity Margolus--Levitin curve), its bandwidth-resolved and multi-observable refinements, and its
backward reading as a model-independent mean-energy certificate. Together these make the
(mean-energy $\times$ observable) corner of the speed-limit dichotomy a complete and operationally grounded
one; its cleanest physical home is the thermodynamics of timekeeping.

\begin{acknowledgments}
I thank L. Maccone for comments on an earlier version.

The author received no funding for this work and declares no competing interests.
\end{acknowledgments}

\medskip\noindent\emph{Use of AI tools.} Anthropic Claude (Opus~4.x/5, via Claude Code, 2026) was used
extensively for copy-editing and rewriting the prose, and for parts of the verification and figure scripts in
the accompanying repository. All output was author-directed and author-verified: numerical claims were checked
against the committed scripts and every reference against its published record. The theorems, proofs, and
constants are the author's own. The author takes full responsibility for all content.

\bibliography{refs}

\end{document}